\documentclass[11pt]{article}

\usepackage[paper=letterpaper, margin=1in]{geometry}
\usepackage{ifthen}

\def\GenerateShortVersion{}
\ifthenelse{\isundefined{\GenerateShortVersion}}
{
\def\LongVersion{}
\def\LongVersionEnd{}
\long\def\ShortVersion#1\ShortVersionEnd{}
}
{
\def\ShortVersion{}
\def\ShortVersionEnd{}
\long\def\LongVersion#1\LongVersionEnd{}
}

\newcommand{\Ignore}[1]{\ignorespaces}

\usepackage{amsmath, amssymb}

\usepackage{amsthm}

\usepackage{ifpdf}

\usepackage{authblk}

\ifpdf
\usepackage[pdftex]{graphicx}
\else
\usepackage[dvips]{graphicx}
\fi

\usepackage{algpseudocode}
\usepackage{algorithm}
\algtext*{EndWhile} % Removes the "end while" text
\algtext*{EndIf} % Removes the "end if" text
\algtext*{EndFor} % Removes the "end for" text

\makeatletter
\renewcommand*{\ALG@name}{Pseudocode}
\makeatother

\usepackage{enumitem}

\LongVersion %{
\LongVersionEnd %}

\ShortVersion %{
\usepackage{times}
\ShortVersionEnd %}

\ShortVersion %{
\renewcommand{\paragraph}[1]{\par\noindent\textbf{#1}}
\ShortVersionEnd %}

\newtheorem{theorem}{Theorem}[section]
\newtheorem{lemma}[theorem]{Lemma}

\theoremstyle{definition}

\theoremstyle{plain}

\LongVersion %{
\newenvironment{MathMaybe}[0]
{\begin{displaymath}\ignorespaces}
{\end{displaymath}}
\LongVersionEnd %}
\ShortVersion %{

\ShortVersionEnd %}

\usepackage{comment}

\usepackage{cite}

\newtheorem{question}[theorem]{Question}

\newcommand{\Sampler}{\ensuremath{\mathtt{Sampler}}}
\newcommand{\Cluster}[1]{\ensuremath{\mathtt{Cluster}_{#1}}}

\renewcommand{\Pr}{\mathbb{P}}

\newcommand{\C}{\mathcal{C}}
\newcommand{\K}{\mathcal{K}}

\newcommand{\Ind}{\ensuremath{\mathrm{Ind}}}
\newcommand{\Dist}{\ensuremath{\mathrm{dist}}}
\begin{document}

\title{Message Reduction in the LOCAL Model is a Free Lunch}
\date{}

\author{Shimon Bitton}
\affil{Technion - Israel Institute of Technology.
\texttt{sbitton@technion.ac.il}}
\author{Yuval Emek}
\affil{Technion - Israel Institute of Technology.
\texttt{yemek@technion.ac.il}}
\author{Taisuke Izumi}
\affil{Nagoya Institute of Technology.
\texttt{t-izumi@nitech.ac.jp}}
\author{Shay Kutten}
\affil{Technion - Israel Institute of Technology.
\texttt{kutten@ie.technion.ac.il}}

\begin{titlepage}

\maketitle

\begin{abstract}
A new \emph{spanner} construction algorithm is presented, working under the
\emph{LOCAL} model with unique edge IDs.
Given an $n$-node communication graph, a spanner with a constant stretch and
$O (n^{1 + \varepsilon})$
edges (for an arbitrarily small constant
$\varepsilon > 0$)
is constructed in a constant number of rounds sending
$O (n^{1 + \varepsilon})$
messages whp.
Consequently, we conclude that every $t$-round LOCAL algorithm can be
transformed into an
$O (t)$-round
LOCAL algorithm that sends
$O (t \cdot n^{1 + \varepsilon})$
messages whp.
This improves upon all previous message-reduction schemes for LOCAL algorithms
that incur a
$\log^{\Omega (1)} n$
blow-up of the round complexity.

\end{abstract}

\noindent

\end{titlepage}

%\clearpage

\pagenumbering{arabic}

%%%%%%%%%%%%%%%%%%%%%%%%%%%%%%%%%%%%%%%%%%%%%%%%%%%%%%%%%%%%%%%%%%%%%%%%%%%%%%
\section{Introduction}
\label{sec:introduction}
%%%%%%%%%%%%%%%%%%%%%%%%%%%%%%%%%%%%%%%%%%%%%%%%%%%%%%%%%%%%%%%%%%%%%%%%%%%%%%
What is the minimum number of messages that must be sent by \emph{distributed
graph algorithm} for solving a certain task?
Is there a tradeoff between the message and time complexities of such
algorithms?
How do the message complexity bounds depend on the exact model assumptions?
These questions are among the most fundamental ones in distributed computing
with a vast body of literature dedicated to their resolution.

A graph theoretic concept that plays a key role in this regard is that of
\emph{spanners}.
Introduced by Peleg and Ullman \cite{PelegU1989} (see also
\cite{peleg1989graph}), an \emph{$\alpha$-spanner}, or a spanner with
\emph{stretch} bound $\alpha$, of a connected graph
$G = (V, E)$
is a (spanning) subgraph
$H = (V, S)$
of $G$ where the distance between any two vertices is at most $\alpha$ times
their distance in $G$.\footnote{%
An equivalent definition requires that it admits a path of length at most
$\alpha$ between any two nodes adjacent in $G$.}
More general spanners, called \emph{$(\alpha, \beta)$-spanners}, are also
considered, where the spanner distance between any two nodes is at most
$\alpha$ times their distance in $G$ plus an additive $\beta$-term
(\cite{ElkinP2004}).

Sparse low stretch spanners are known to provide the means to save on message
complexity in the \emph{LOCAL} model
\cite{Linial1992, Peleg2000}
without a significant increase in the round complexity.
This can be done via the following classic simulation technique:
Given an $n$-node communication graph
$G = (V, E)$
and a LOCAL algorithm $\mathcal{A}$ whose run $\mathcal{A}(G)$ on $G$ takes
$t$ rounds,
(1)
construct an $\alpha$-spanner
$H = (V, S)$
of $G$;
and
(2)
simulate each communication round of $\mathcal{A}(G)$ by $\alpha$
communication rounds in $H$ so that a message sent over the edge
$(u, v) \in E$
under $\mathcal{A}(G)$ is now sent over a
$(u, v)$-path
of length at most $\alpha$ in $H$.
The crux of this approach is that the simulating algorithm executed in stage
(2) runs for
$\alpha t$
rounds and sends at most
$2 \alpha t \cdot |S|$
messages.
Therefore, if $\alpha$ and $|S|$ are 'small', then the simulating algorithm
incurs 'good' round and message bounds.
In particular, the performance of the simulating algorithm does not depend on
the number $|E|$ of edges in the underlying graph $G$.

What about the performance of the spanner construction in the 'preprocessing'
stage (1) though?
A common thread among distributed spanner construction algorithms is that they
all send
$\Omega (|E|)$
messages when running on graph
$G = (V, E)$.
Consequently, accounting for the messages sent during this preprocessing stage,
the overall message complexity of the aforementioned simulation technique
includes a seemingly inherent
$\Omega (|E|)$
term.
The following research question that lies at the heart of distributed message
reduction schemes is therefore left open.

\begin{question} \label{question:message-reduction}
Given a LOCAL algorithm $\mathcal{A}$ whose run $\mathcal{A}(G)$ on $G$ takes
$t$ rounds, is it possible to simulate $\mathcal{A}(G)$ in
$O (t)$
rounds while sending only
$O (n^{1 + \varepsilon})$
messages for an arbitrarily small constant
$\varepsilon > 0$,
irrespective of the number $|E|$ of edges in $G$?
\end{question}
This question would be resolved on the affirmative if one could design a LOCAL
algorithm that constructs an $\alpha$-spanner
$H = (V, S)$
of $G$ with stretch
$\alpha = O (1)$
and
$|S| = O (n^{1 + \varepsilon})$
edges in
$O (1)$
rounds sending
$O (n^{1 + \varepsilon})$
messages.
Despite the vast amount of literature on distributed spanner construction
algorithms
\cite{BaswanaS2007, elkin2007near, DerbelGPV2008, DerbelG2006, DerbelGPV2009,
Pettie2010, ElkinZ2006, ElkinN2017},
it is still unclear if such a LOCAL spanner construction algorithm exists.

Some progress towards the positive resolution of
Question~\ref{question:message-reduction} has been obtained by Censor-Hillel
et al.~\cite{Censor-HillelHKM2012} and Haeupler~\cite{Haeupler2015} who
introduced techniques for simulating LOCAL algorithms by \emph{gossip}
processes.
Using this approach, one can transform any $t$-round LOCAL
algorithm into a LOCAL algorithm that runs in 
$O (t \log n + \log^{2} n)$
rounds while sending $n$ messages per round \cite{Haeupler2015}.
This transformation provides a dramatic message complexity improvement if one
is willing to accept algorithms that run for
$\log^{O (1)} n$
many rounds, e.g., if the the bound $t$ on the
round complexity of the original algorithm is already
in the
$\log^{O (1)} n$
range.
However, if
$t = \log^{o (1)} n$,
then the gossip based message reduction scheme of
\cite{Censor-HillelHKM2012, Haeupler2015}
significantly increases the
round complexity and this increase seems to be inherent to that technique.

%%%%%%%%%%%%%%%%%%%%%%%%%%%%%%%%%%%%%%%
\subsection{Definitions and Results}
\label{sec:results}
%%%%%%%%%%%%%%%%%%%%%%%%%%%%%%%%%%%%%%%
Throughout, we consider a communication network represented by a connected
unweighted undirected graph
$G = (V, E)$
and denote
$n = |V|$.
The nodes of $G$ participate in a distributed algorithm under the (fully
synchronous) \emph{LOCAL} model
\cite{Linial1992, Peleg2000}
with the following two model
assumptions:
(i)
the nodes know an $O(1)$-approximate upper bound on $\log n$ (equivalently,
a $\mathrm{poly(n)}$-approximate upper bound of $n$) at all times;
and
(ii)
the graph admits \emph{unique edge IDs} so that the ID of an edge is known
to both its endpoints at all times.\footnote{%
Alternatively, the algorithm can run under the rather common $KT_{1}$ model
variant \cite{AwerbuchGVP1990}, where the nodes are associated with unique IDs
and each node knows the ID of the other endpoint of each one of its incident
edges;
see the discussion in Section~\ref{sec:related-work}.}
Other than that, the nodes have no a priori knowledge of $G$'s topology.
Our main technical contribution is a new algorithm for constructing sparse
spanners, called \Sampler{}, whose guarantees are cast in the following
Theorem.

\begin{theorem} \label{thm:main}
Fix integer parameters
$1 \leq k \leq \log\log n$
and
$0 \leq h \leq \log n$.
Algorithm \Sampler{} constructs an edge set
$S \subseteq E$
of size
$|S| \leq \tilde{O} (n^{1 + 1 / (2^{k+1} - 1)})$
such that
$H = (V, S)$
is an
$O(3^k)$-spanner of $G$
whp.\footnote{%
We say that an event occurs \emph{with high probability}, abbreviated by
\emph{whp}, if the probability that it does not occur is at most $n^{-c}$ for
an arbitrarily large constant $c$.}\,\footnote{%
The asymptotic notation
$\tilde{O} (\cdot)$
may hide
$\log^{O (1)} n$
factors.}\,%
The round complexity of \Sampler{} is
$O (3^k h)$
and its message complexity is
$\tilde{O} (n^{1 + 1 / (2^{k+1} - 1) + (1/h)})$
whp.
\end{theorem}

By setting the parameters $k$ and $h$ so that
$1 / (2^{k + 1} - 1) = 1 / h = \varepsilon / 2$
for an arbitrarily small constant
$\varepsilon > 0$
and utilizing the aforementioned spanner based simulation technique, we obtain
a message-reduction scheme that transforms any LOCAL algorithm $\mathcal{A}$
whose run on $G$ takes $t$ rounds into a (randomized) LOCAL algorithm that
runs in
$O (t)$
rounds and sends
$\tilde{O} (t n^{1 + \varepsilon})$
messages whp.
This resolves Question~\ref{question:message-reduction} on the affirmative
provided that one is willing to tolerate a
$1 / \mathrm{poly}(n)$
error probability.
In fact, we can improve the message reduction scheme even further via the
following two-stage process:
first, use the $\alpha$-spanner
$H = (V, S)$
constructed by \Sampler{} to simulate the run on $G$ of some off-the-shelf
LOCAL algorithm that constructs an $\alpha'$-spanner
$H' = (V, S')$
with a better tradeoff between $\alpha'$ and $|S'|$;
then, use $H'$ to simulate the run of $\mathcal{A}$ on $G$.
In Section~\ref{sec:simulation}, we show that with the right choice of
parameters, this two-stage process leads to the following theorem.

\begin{theorem} \label{thm:simulation}
  Every distributed task solvable by a $t$-round LOCAL algorithm can be solved
  with
  %shay (arbitrary)
  any one of the following pairs of time and message complexities:
  \begin{itemize}
    \item $\tilde{O}(tn^{1 + 2/(2^{\gamma+1} - 1)})$ message complexity and
    $O(3^\gamma t + 6^\gamma)$ round complexity for any
    $1 \leq \gamma \leq \log\log n$,
    \item $\tilde{O}(t^2n^{1 + O(1/\log t)})$ message
    complexity and $O(t)$ round complexity.
    \end{itemize}
\end{theorem}

%%%%%%%%%%%%%%%%%%%%%%%%%%%%%%%%%%%%%%%
\subsection{Related Work and Discussion}
\label{sec:related-work}
%%%%%%%%%%%%%%%%%%%%%%%%%%%%%%%%%%%%%%%

%%%%%%%%%%%%%%%%%%%%%%%%%%%%%%%%%%%%%%%
\paragraph*{Model Assumptions}
%%%%%%%%%%%%%%%%%%%%%%%%%%%%%%%%%%%%%%%
The current paper considers the fully synchronous message passing LOCAL model
\cite{Linial1992, Peleg2000}
that ignores the message size and focuses only on locality considerations.
This model has been studied extensively (at least) since the seminal
paper of Linial \cite{Linial1992}, with special attention to the question of
what can be computed efficiently, including some recent interesting
developments, see, e.g., the survey in
\cite[Section 1]{ghaffari-complexity-local}.
The more restrictive \emph{CONGEST} model \cite{Peleg2000}, where message size
is bounded (typically to $O (\log n)$ bits), has also been extensively
studied.

Many variants of the LOCAL model have been addressed over the years,
distinguished from each other by the exact model assumptions, the most common
such assumptions being unique node IDs and knowledge of $n$.
Another important distinction addresses the exact knowledge held by any node
$v$ regarding its incident edges when the execution commences.
Two common choices in this regard are the $KT_{0}$ variant, where $v$
knows only its own degree, and the $KT_{1}$ variant, where $v$ knows the ID of
$e$'s other endpoint for each incident edge $e$ \cite{AwerbuchGVP1990}.
The authors of \cite{AwerbuchGVP1990} advocate
 $KT_1$, arguing that it is
the more natural among the two model variants, but papers have been published
about each of them.

In the current paper, it is assumed that each \emph{edge}
$(u, v) \in E$
is equipped with a unique ID, known to both $u$ and $v$.
In general, this assumption
%shay: implies unique node IDs and
%shay: (it is implied by unique node IDs.)
lies (strictly) between the $KT_{0}$ and $KT_{1}$
model variants.
Note that the unique edge IDs assumption is no longer weaker than the $KT_{1}$
assumption when the communication graph admits \emph{parallel} edges. However,
our algorithm and analysis apply also to such graphs (assuming that
$|E| \leq n^{O (1)}$)
under either of the two assumptions.

%%%%%%%%%%%%%%%%%%%%%%%%%%%%%%%%%%%%%%%
\paragraph*{Message Complexity $o(|E|)$}
%%%%%%%%%%%%%%%%%%%%%%%%%%%%%%%%%%%%%%%
As discussed in Section~\ref{sec:results}, the main conceptual contribution of
this paper is that on graphs with
$m = |E| \gg n$
edges, many distributed tasks can now be solved by sending
$o (m)$
messages while keeping the round complexity unharmed.
The challenge of reducing the message complexity below
$O (m)$
has already received significant attention.
In particular, it has been proved in \cite{KuttenPPRT2015} that under the
CONGEST $KT_0$ model, intensively studied \emph{global} tasks, namely,
distributed tasks that require
$\Omega (D)$
rounds, where $D$ is the graph's diameter (e.g., broadcasting, leader
election, etc.), cannot be solved unless
$\Omega (m)$
messages are sent (in the worst case).

This is no longer true under more relaxed models.
For example, under the LOCAL $KT_{1}$ model, DFS and leader election can be
solved by sending
$O (n)$
and
$O (n \log n)$
messages, respectively \cite{korach1990modular}.
This implies similar savings in the number of messages required for most
global tasks (trivially, by collecting all the information to the leader).
Under the CONGEST $KT_{1}$ model, it has been recently proved that a
\emph{minimum spanning tree} can be constructed, sending
$o (m)$
messages \cite{king2015construction,GmyrP18,ghaffari2018distributed,mashreghi2018broadcast}.
%%%%%%%%%%%%%%% Insert citation!!!

Restricted graph classes have also been addressed in this regard.
In particular, the authors of \cite{kutten2013sublinear} proved that under the
CONGEST $KT_{0}$ model, the message complexity of leader election is
$O(\sqrt{n} \log^\frac{3}{2} n)$
whp in the complete graph and more generally,
$O(\tau \sqrt{n} \log^\frac{3}{2} n)$
whp in graphs with mixing time $\tau(G)$.

%%%%%%%%%%%%%%%%%%%%%%%%%%%%%%%%%%%%%%%
\paragraph*{Spanners}
%%%%%%%%%%%%%%%%%%%%%%%%%%%%%%%%%%%%%%%
%shay: already said: Since their introduction in \cite{PelegU1989} (see also
%\cite{peleg1989graph}),
Graph spanners have been extensively studied and
papers dealing with this fundamental graph theoretic object are too numerous
to cite.
Beyond the role that sparse spanners play in reducing the message complexity
of distributed (particularly LOCAL) algorithms as discussed in
Section~\ref{sec:results}, spanners have many applications in various
different fields, some of the more relevant ones include
synchronization
\cite{awerbuch1985complexity, PelegU1989},
routing
\cite{peleg1989trade, awerbuch1998near},
and distance oracles \cite{thorup2005approximate}.

Many existing distributed spanner algorithms
% rely on `distance-sensitive
%computations' over some $r$-hop neighborhood,
%e.g.,
have a node collect the topology of the graph up to a distance of some $r$ from itself
\cite{DubhashiMPRS2005, DerbelG2006}
or employ more sophisticated bounded diameter graph decomposition
techniques involving the node's neighborhood
\cite{Elkin2005, ElkinZ2006, MillerPVX2015, ElkinN2017}
such as the techniques presented in
\cite{LinialS1993, BlellochGKMPT2014, ElkinN2016}.
This approach typically requires sending messages over every edge at distance
at most $r$ from some subset of the nodes which leads to a large number of
messages.
Another approach to constructing sparse spanners in a distributed manner is to
recursively grow local clusters
\cite{BaswanaS2007, DerbelGPV2008, Pettie2010,barenboim2018fast,censor2018distributed}.
Although this approach does not require the (explicit) exploration of
multi-hop neighborhoods, the existing algorithms operating this way also admit
large message complexity because too many nodes have to explore their $1$-hop
neighborhoods.
Our algorithm \Sampler{} is inspired by the algorithm of
\cite{BaswanaS2007} and adheres to the latter approach, but it is
designed in a way that drastically reduces the message complexity --- see
Section~\ref{sec:techniques-overview}.

%Over the years, several useful variants of the \emph{purely multiplicative}
%spanners discussed so far have been introduced.
%These include
%\emph{$\beta$-additive spanners}
%\cite{AingworthCIM1999},
%where the spanner distance between every two nodes $x$ and $y$ is at most
%their distance in the original graph $G$ plus an additive $\beta$-term,
%more general
%\emph{$(\alpha, \beta)-spanners$}
%\cite{ElkinP2004},
%where the spanner distance between $x$ and $y$ is at most
%$\alpha$
%times their distance in $G$ plus an additive $\beta$-term.
%Finally, the study of (multiplicative) spanners in weighted graphs has been
%initiated in \cite{AlthoeferDDJS1993}.
%In this context, one is typically interested in constructing low-stretch
%spanners, bounding both the number of edges and their total weight.

%%%%%%%%%%%%%%%%%%%%%%%%%%%%%%%%%%%%%%%
\subsection{Techniques Overview}
\label{sec:techniques-overview}
%%%%%%%%%%%%%%%%%%%%%%%%%%%%%%%%%%%%%%%
Algorithm \Sampler{} employs \emph{hierarchical node sampling}, where a
sampled node $u$ in level
$j$
of the hierarchy forms the \emph{center} of a cluster that includes (some of)
its non-sampled neighbors $v$.
An edge connecting $u$ and $v$ is added to the spanner.
The clusters are then contracted into the nodes of the next level $j+1$.
Also added to the spanner are all
incident edges of every non-center node that has no adjacent center.

This hierarchical node sampling is used also in the distributed spanner
construction of Baswana and Sen~\cite{BaswanaS2007} and similar recursive
clustering techniques were used in other papers as well (see
Section~\ref{sec:related-work}).
Common to all these papers is that the centers in each level communicate
directly with \emph{all} their neighbors to facilitate the cluster forming
task.%
\footnote{Some of these papers consider weighted communication graphs and as
such, have to deal with other issues that are spared in the current paper.}
As this leads inherently to
$\Omega (|E|)$
messages, we are forced to follow a different approach, also based on (a different) sampling process.

The first thing to notice in this regard is that it is enough for a non-center
node $v$ to find a \emph{single} center $u$ in its neighborhood, so perhaps
there is no need for the center nodes to announce their status to \emph{all}
their neighbors?
Indeed, we invoke an \emph{edge sampling} process in the non-center nodes $v$
that identifies a subset of $v$'s incident edges over which \emph{query}
messages are sent.
Our analysis shows that
(1)
the number of query messages is small whp;
(2)
if $v$ does have an adjacent center, then the edge sampling process finds such a
center whp;
and
(3)
if $v$ does not have an adjacent center, then $v$'s degree is small; in this case, all its
incident edges are queried and join the spanner whp.

However, this edge sampling idea by itself does not suffice:
Note that the graph (in some level of the hierarchical construction) is constructed via cluster contraction (into a single node) in lower levels of the
hierarchy. Hence this graph typically exhibits edge multiplicities (even if the original
communication graph is simple).
This means that some neighbors $w$ of $v$ may have many more (parallel)
$(w, v)$-edges
than others. Informally, this can bias the probabilities for finding additional neighbors in the edge sampling process.
The key idea in resolving this issue is to run the edge sampling process (in each level) in a
carefully designed \emph{iterative} fashion. Intuitively, the first iterations in each level ``peels off'' the neighbors to which a
large fraction of $v$'s incident edges lead.
This increases the probability of finding one of the rest of the neighbors in later iterations.
Note that once $v$ found a neighbor $u$, node $v$ can identify all $v$'s edges leading to $u$ (and so ``peel them off'' from the next iterations), by having $u$ report to $v$ the IDs of all the edges touching $u$.

%%%%%%%%%%%%%%%%%%%%%%%%%%%%%%%%%%%%%%%
\subsection{Paper Organization}
\label{sec:paper-organization}
%%%%%%%%%%%%%%%%%%%%%%%%%%%%%%%%%%%%%%%
The rest of the paper is organized as follows.
Following some preliminary definitions provided in
Section~\ref{sec:preliminaries}, \Sampler{} is presented in
Section~\ref{sec:algorithm} and analyzed in Section~\ref{sec:analysis}.
For clarity of the exposition, we first present \Sampler{} as a centralized
algorithm and then, in Section~\ref{sec:distr-impl}, explain how it can be
implemented under the LOCAL model with the round and message complexities
promised in Theorem~\ref{thm:main}.
Section~\ref{sec:simulation} describes how   algorithm \Sampler{} is used
to obtain the message-reduction schemes we mentioned in the introduction.
Finally, we conclude the paper in Section~\ref{sec:concludingremarks}.

%%%%%%%%%%%%%%%%%%%%%%%%%%%%%%%%%%%%%%%%%%%%%%%%%%%%%%%%%%%%%%%%%%%%%%%%%%%%%%
\section{Preliminaries}
\label{sec:preliminaries}
%%%%%%%%%%%%%%%%%%%%%%%%%%%%%%%%%%%%%%%%%%%%%%%%%%%%%%%%%%%%%%%%%%%%%%%%%%%%%%
Consider some graph
$G = (V, E)$.
When convenient, we may denote the node set $V$ and edge set $E$ by
$V(G)$ and $E(G)$, respectively.
Unless stated otherwise, the graphs considered throughout this paper are
undirected and not necessarily simple, namely, the edge set $E$ may include
edge multiplicities (a.k.a.\ \emph{parallel} edges).
Given disjoint node subsets
$U, U' \subseteq V$,
let $E(U)$ denote the subset of edges with (exactly) one endpoint in $U$ and
let
$E(U, U') = E(U) \cap E(U')$
denote the subset of edges with one endpoint in $U$ and the
other in $U'$.
If
$U = \{ u \}$
and
$U' = \{ u' \}$
are singletons, then we may write $E(u)$ and $E(u, u')$ instead of
$E(\{ u \})$
and
$E(\{ u \}, \{ u' \})$,
respectively
(notice that
$E(u, u')$
may contain multiple edges when $G$ is not a simple graph).

Let
$\C = \{ C_{1}, \dots, C_{\ell} \}$
be a collection of non-empty pairwise disjoint node subsets referred to as
\emph{clusters} of $G$.\footnote{%
Notice that the union of the clusters is not required to be the
whole node set $V$. 
}
The \emph{cluster graph} (cf.~\cite{Peleg2000}) induced by $\C$ on $G$,
denoted by $G(\C)$, is the undirected graph whose nodes are identified with
the clusters in $\C$, and the edges connecting nodes $C_{i}$
and $C_{j}$,
$1 \leq i \neq j \leq \ell$,
correspond to the edges crossing between clusters $C_{i}$ and $C_{j}$ in
$G$, that is, the edges in
$E(C_{i}, C_{j})$.
Observe that $G(\C)$ may include edge multiplicities even if $G$ is a
simple graph. We denote by $\Ind_G(C)$ the subgraph of $G$ induced by $C$. When convenient, the term ``cluster $C$ applies also to $\Ind_G(C)$\footnote{Each cluster here will be a connected component.}
For $u, v, \in V$, the distance between $u$ and $v$ in $G$ is denoted by
$\Dist_G(u, v)$.

As stated in the introduction, the assumption about global parameters is
that each node knows an $O(1)$-approximate upper bound on
$\log n$. For the sake of simplicity, we treat the algorithm as if each node knows the exact value of
$\log n$, however, this is not essential.

%%%%%%%%%%%%%%%%%%%%%%%%%%%%%%%%%%%%%%%%%%%%%%%%%%%%%%%%%%%%%%%%%%%%%%%%%%%%%%
\section{Constructing an $O(3^k)$-Spanner}
\label{sec:algorithm}
%%%%%%%%%%%%%%%%%%%%%%%%%%%%%%%%%%%%%%%%%%%%%%%%%%%%%%%%%%%%%%%%%%%%%%%%%%%%%%

In the following argument, let $\delta = 1 / (2^{k+1} - 1)$ and $\epsilon = 1/h$
for short. Denote the simple graph input to the algorithm by
$G_{0} = (V_{0}, E_{0})$.
Algorithm \Sampler{} (see Pseudocode~\ref{pseudocode:spanner}) generates a
sequence
$G_{1}, \dots, G_{k}$
of graphs, where
$G_{j} = (V_{j}, E_{j})$.
Let $n_j$ and $m_j$ be the numbers of nodes and edges in $G_j$ respectively,
$N_{j}(v)$ be the set of neighbors in $G_{j}$ of node $v \in V_{j}$,
and $E_j(v, u)$ be the set of edges connecting $v$ to $u$ in $G_j$.
The process is executed in an iterative fashion, where in each iteration
$j = 0, \dots, k$,
the algorithm constructs a collection
$\C \subset 2^{V_{j}}$
of pairwise disjoint clusters and an edge set
$F \subseteq E_{j}$
that is added to the spanner edge set $S$ (as an exception, in the final iteration
of $j = k$, only $F$ is constructed but the cluster collection $\C$ is not created).
The graph $G_{j+1}$ is then defined to be the cluster graph
$G_{j}(\C)$
induced by $\C$ on $G_{j}$.
The construction of $\C$ and $F$ is handled by procedure \Cluster{j}
that is described soon.
To avoid confusion, in what follows, we fix
$n = n_0 = |V_{0}|$.

\begin{algorithm}
  \caption{\label{pseudocode:spanner}%
  \Sampler{}.}
  \begin{algorithmic}[1]
  \State{$n \leftarrow |V_{0}|$; $\delta \leftarrow 1 / (2^{k + 1} - 1)$;
  $S \leftarrow \emptyset$}
  \For{$j = 0, \dots, k$}
    \State{$\langle \C, F \rangle \leftarrow \Cluster{j}$}
    \State{$S \leftarrow S \cup F$}
    \If{$j < k$}
      \State{$G_{j+1} \leftarrow G_{j}(\C)$}
    \EndIf
  \EndFor
  \State{return $S$}
  \end{algorithmic}
\end{algorithm}

%%%%%%%%%%%%%%%%%%%%%%%%%%%%%%%%%%%%%%%
\paragraph*{Procedure \Cluster{j}}
%%%%%%%%%%%%%%%%%%%%%%%%%%%%%%%%%%%%%%%
$\;\;$ On input graph
$G_{j}$ ($0 \leq j \leq k$),
this procedure constructs the cluster collection
$\C \subset 2^{V_{j}}$
and edge set
$F = \cup_{v \in V_j} F_v$ ($F_v \subseteq E_{j}$) to be added to the spanner edges.
The procedure (see Pseudocode~\ref{pseudocode:cluster})
consists of two steps. In the first step, each node $v$ tries to identify
%shay: was written: $\min\{c\log n(\exp_{n}(2^{j} \delta), |N_j(v)|\}$
$\min\{c\exp_{n}(2^{j} \delta)\log n, |N_j(v)|\}$
neighbors by an iterative random-edge sampling process\footnote{For ease of writing long exponents, define
$\exp_a(x) = a^x$.}, where $c$ is a sufficiently
large constant to guarantee high success probability of the algorithm.
For this process, $v$ maintains a set $X_v \subseteq E_{j}(v)$ of the
edges that have not been explored yet. Initially, $X_v$ is set to $X_v = E_{j}(v)$;
the content of $X_v$ is then gradually eliminated by running
$2/\epsilon = 2h$ \emph{trials}. In every trial, each node $v \in V_{j}$
chooses $c^2 \exp_{n} (2^{j} \delta + \epsilon) \log^3 n$
edges from $X_v$ independently and uniformly at random (possibly choosing
the same edge twice or more). Each chosen edge is said to be
a \emph{query} edge. For a neighbor $u \in N_{j}(v)$ such that
$E_j(v, u)$ contains a query edge, we say that $u$ is \emph{queried} by
$v$. For each queried node $u \in N_{j}(v)$, $v$ adds one arbitrary
query edge $e \in E_{j}(u, v)$ to the edge set $F_v$,
and eliminates all the edges in $E_{j}(u, v)$ from $X_v$ (Figure~\ref{fig:cluster} (a)-(c)).
Then the procedure advances to the next trial unless
$|F_v| \geq c\exp_{n}(2^{j} \delta)\log n$ holds or $X_v$ is emptied.
Let $\hat{N}_{j}(v) \subseteq N_{j}(v)$ be the set of the nodes queried
by $v$ after finishing $2h$ trials. A node $v \in V_{j}$ is called
\emph{light} if $\hat{N}_j(v) = N_j(v)$ holds, or called
\emph{heavy} if $|N_j(v)| > |\hat{N}_j(v)| \geq c\exp_{n} (2^{j} \delta)\log n$. It is proved later that every node becomes either light or heavy whp.

In the second step, the algorithm creates (only if $j<k$) the vertex set $V_{j+1}$ 
by clustering the nodes in $V_{j}$. 
%Note that this step is not executed
%during the exection of $\Cluster{k}$ (i.e., it is executed in $\Cluster{j}$ only for
%$j \in [0, k-1]$).
%This step is executed for
%$j = 0, 1, \dots, k-1$ (not including $k$). 
The algorithm marks
each node $v \in V_{j}$ as a \emph{center} w.p. $p_{j} =
\exp_{n} (-2^{j} \delta)$. Each node $v$ having a center $u$ contained
in $\hat{N}_{j}(v)$ is merged into $u$ (if two or more centers are contained,
an arbitrary one is chosen). A merged cluster corresponds to a node in $G_{j+1}$
(Figure~\ref{fig:cluster} (d)-(f)).
As a result, letting $\C \subset 2^{V_{j}}$ be the clusters inducing
the cluster graph $G_{j+1}$, each cluster $C = C(u) \in \C$ contains
exactly one center $u \in V_{j}$ and some subset of $N_{j}(u)$.
The node which is not
merged into any center is said to be an \emph{unclustered} node.
Due to some technical reason, all the nodes in $G_k$ are defined to be
unclustered. It is shown that every heavy node is merged into some center whp.,
and that every node in $G_k$ is light (proved later). Thus, every unclustered
node is light.

  \begin{algorithm}
    \caption{\label{pseudocode:cluster}%
    $\Cluster{j}$.}
    \begin{algorithmic}[1]
    \ForAll{$v \in V_{j}$}
      \State{$F'_v \leftarrow \emptyset$; $F_v \leftarrow \emptyset$;
      $X_v \leftarrow E_{j}(v)$; $i \leftarrow 1$}
      \State{/$\ast$ First Step $*$/}
      \While{$(i \leq 2h) \ \wedge \  (|F'_v| < c\exp_n(2^{j}\delta)\log n) \ \wedge \
      (|X_v| \neq \emptyset)$}%
      \Comment{run (at most) $2h$ trials}
        \For{$x=1$ to $c^2\exp_{n}(2^{j}\delta + \epsilon)\log^3 n $}
          \State{add an edge $e$ selected uniformly at random from $X_v$ to $F'_v$}
        \EndFor
        \While{$(F'_v \setminus F_v) \neq \emptyset$}
        \Comment{$F'_v \setminus F_v$ is the set of edges newly added in the current trial}
          \State{Pick an arbitrary edge $e = (v, u) \in F'_v \setminus F_v$}
          \State{Remove all the edges incident to $u$ from $X_v$}
          \State{Remove all the edges incident to $u$ other than $e$ from $F'_v$}
          \State{$F_v \leftarrow F_v \cup \{ e \}$}
        \EndWhile
        \State{$i \leftarrow i + 1$}
      \EndWhile
    \EndFor
    \State{$F \leftarrow \cup_{v \in V_j} F_v$}
    \State{/$\ast$ Second Step $*$/}
    \If{$j < k$}
      \ForAll{$v \in V_{j}$}
      \State{mark $v$ as a center and create $C(v) = \{v\}$ w.p. $\exp_{n} (-2^{j} \delta)$}
      \EndFor
      \ForAll{non-center $v \in V_{j}$}
        \If{$\exists (v, u) \in F : u$ is a center}%
          \State{$C(u) \leftarrow C(u) \cup \{v\}$}
        \EndIf
      \EndFor
    \EndIf
    \State{return $\langle \{ C(u) \mid u \text{ is a center} \}, F \rangle$}
    \end{algorithmic}
    \end{algorithm}

%%%%%%%%%%%%%%%%%%%%%%%%%%%%%%%%%%%%%%%%%%%%%%%%%%%%%%%%%%%%%%%%%%%%%%%%%%%%%%
\section{Analysis}
\label{sec:analysis}
%%%%%%%%%%%%%%%%%%%%%%%%%%%%%%%%%%%%%%%%%%%%%%%%%%%%%%%%%%%%%%%%%%%%%%%%%%%%%%
Throughout this section, we refine the definition of terminology ``whp.'' to claim
that the probabilistic event considered in the context holds with probability
$1 - 1/n^{\Theta(c)}$ for parameter $c$ defined in the algorithm. 
With a small abuse of probabilistic arguments,
we treat those events as if they necessarily occur (with probability one).
Since we only handle a polynomially-bounded number of probabilistic events
in the proof, the standard union-bound argument ensures that any consequence
of the analysis also holds whp. for a sufficiently large $c$.
We begin the analysis by bounding the number of nodes in graph $G_{j}$.
Denote $\hat{p}_j = \prod_{0 \leq i \leq j} p_i$. It is easy to check that
$\hat{p}_{j - 1} = \exp_{n}(-(2^{j} - 1)\delta) $ for $1 \leq j \leq k$.
\begin{lemma} \label{lem:bound-graph-size}
  Graph $G_{j}$ satisfies
  $n\hat{p}_{j-1}/2 \leq n_j \leq 3n\hat{p}_{j-1}/2$
  %$\exp_{n} (1 - (2^{j} - 1) \delta) / 2
  %\leq
  %n_{j}
  %\leq
  %2\exp_{n} (1 - (2^{j} - 1) \delta)$
  whp. for $1 \leq j \leq k$.
  \end{lemma}
  \begin{proof}
    The value $n_j$ follows the binomial distribution of $n_{j-1}$ trials and success
    probability $p_{j-1}$. Applying Chernoff bound (under conditioning $n_{j-1}$),
    the inequality below holds for all $1 \leq j \leq k$ whp.\footnote{Chernoff bound
    for the binomial distribution $X$ of $m$ trials and success probability $p$ is
    $\Pr[|X - mp|\geq \alpha mp] \leq 2e^{-\frac{\alpha^2 mp}{3}}$.}
    \[
      \left(1 - \sqrt{\frac{c \log n}{n_{j-1}p_{j-1}}}\right)n_{j-1}p_{j-1}
      \leq
      n_j
      \leq
      \left(1 + \sqrt{\frac{c \log n}{n_{j-1}p_{j-1}}}\right)n_{j-1}p_{j-1}.
    \]
    By the definition, $n_j$ and $p_j$ are both non-increasing with respect to $j$.
    Hence we have
    \begin{align*}
      \label{eq:boundnumber}
      \left(1 - \sqrt{\frac{c \log n}{n_{j-1}p_{j-1}}}\right)^{j-1} n\hat{p}_{j-1}
      \leq
      n_j
      \leq
      \left(1 + \sqrt{\frac{c \log n}{n_{j-1}p_{j-1}}}\right)^{j-1} n\hat{p}_{j-1}.
    \end{align*}
    \sloppy{
    We prove the lemma inductively (on $j$) following this inequality.
    For $j = 1$, $n_{j-1} = n$ and thus $|\sqrt{c \log n / n_{j-1}p_{j-1}}| \leq 1/2$
    holds. From the inequality above, we obtain $n_0p_0/2 \leq n_1 \leq 3n_0p_0/2$.
    Suppose $n\hat{p}_{j-2}/2 \leq n_{j-1} \leq 3n\hat{p}_{j-1}/2$. Since
    $n\hat{p}_{j - 1} = \exp_{n}(1 - (2^{j} - 1)\delta) \geq n^{1/2}$ holds for
    $j \leq k$, we have $\sqrt{c \log n / (n_{j - 1}p_{j-1})} \leq
    c\log n / n^{1/2} \ll 1/2(j-1)$ for sufficiently large $n$.
    Applying the approximation of $(1 + x)^y \approx 1 + xy$
    for $|x| \ll 1$ to the inequality, we obtain the lemma.
    }
  \end{proof}
We next prove the facts mentioned in the explanation of
the procedure \Cluster{j}.
\begin{lemma} \label{lem:heavynode}
  For any $0 \leq j \leq k-1$, any heavy node $v \in V_{j}$ contains
  at least one center in $\hat{N}_{j}(v)$ whp.
\end{lemma}

\begin{proof}
\sloppy{
The probability that no center is contained in $\hat{N}_{j}$ is
$(1 - p_{j})^{|\hat{N}_j(v)|}$. Since $|\hat{N}_j(v)| \geq
c \exp_{n}(2^{j}\delta) \log n = c\log n/p_{j}$ holds for any heavy node $v$,
the probability is at most $1/n^c$.
}
\end{proof}

\begin{figure}[t]
  \centering
  \includegraphics[width=110mm,keepaspectratio]{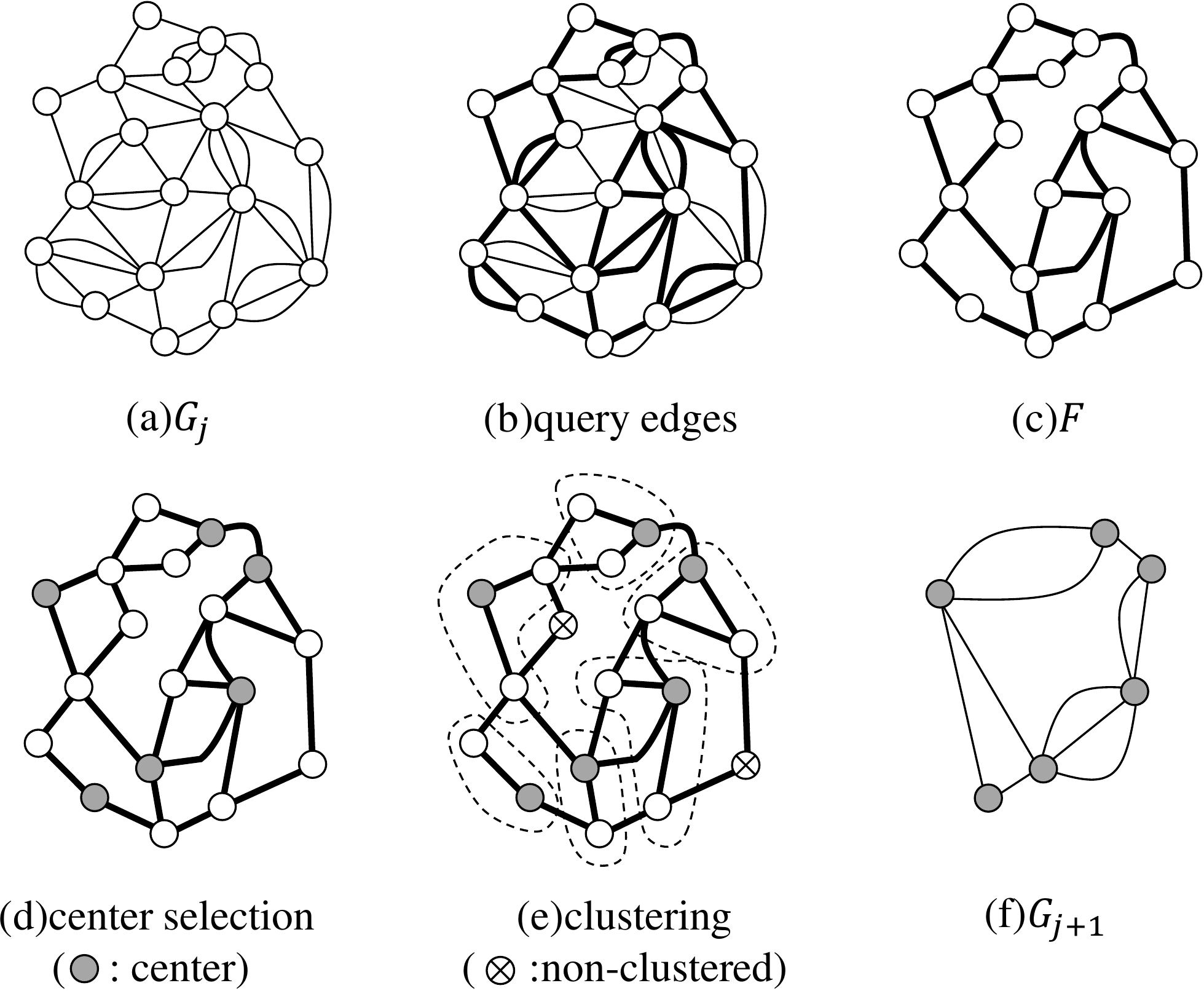}
  \caption{Procedure \Cluster{j}}
  \label{fig:cluster}
\end{figure}

\begin{lemma} \label{lem:unclustered}
For any $0 \leq j \leq k$, any node $v \in V_{j}$ becomes light or heavy whp.
Furthermore, any node $v \in V_{k}$ becomes light whp.
\end{lemma}

\begin{proof}
  Let $\alpha = (3c \exp_n(2^{j}\delta) \log^2 n) /h$ for short.
  For $W \subseteq N_{j}(v)$, define
  $E_{j}(v, W) = \bigcup_{u \in W} E_{j}(v, u)$.
  Let $N^i_{j}(v) \subseteq N_{j}(v)$ be the set of the nodes not
  queried by $v$  at the beginning of the $i$-th trial, and $m_i$ be the numbers
  of edges in $X_v$ at the beginning of the $i$-th trial.
  For any node $u \in N^i_j(v)$, the value of $|E_{j}(v, u)|$
  is called the \emph{volume} of $u$. Similarly, for any $X \subseteq N^i_j(v)$,
  we call the value of $|E_{j}(v, X)|$ the volume of $X$.
  Divide $N^i_{j}(v)$ into $2/\epsilon = 2h$ classes: A node
  $u \in N^i_{j}(v)$ belongs to the $x$-th class $\K^i_x(v)$ if
  $|E_{j}(v, u)| \in (\exp_{n}(x\epsilon), \exp_{n}((x+1)\epsilon)]$ holds
  ($0 \leq x \leq 2h - 1$). Let $\K^i(v)$ be the maximum-volume class
  of all at the beginning of the $i$-th trial. Since
  the volume of $\K^i(v)$ is at least $m_i/2h$, there exists
  at least one node $u \in \K^i(v)$ satisfying
  $|E_{j}(v, u)| \geq m_i/(2h|\K^i(v)|)$. By the definition of class $\K^i(v)$,
  it implies that the volume of any node $u \in \K^i(v)$ is at least
  $m_i/(2h|\K^i(v)|{n^{\epsilon}})$.

  Let $\beta$ be the non-negative integer satisfying $(\beta - 1)\alpha
  \leq |K^i(v)| \leq \beta \alpha$. Then we consider an arbitrary partition
  of $\K^i(v)$ into $q = \lfloor |\K^i(v)|/\beta \rfloor$ groups
  $K_1, K_2, \dots, K_q$ of size $\beta$. Note that $\beta q$ is not
  necessarily equal to $|\K^i(v)|$, but the residuals are omitted.
  Since any node in $\K^i(v)$ has a volume at least
  $m_i/(2h|\K^i(v)|{n^{\epsilon}})$, the volume of $K_{\ell}$ ($1 \leq \ell \leq q$)
  is at least $\beta m_i/(2h|\K^i(v)|{n^{\epsilon}})$. Thus
  the probability that a query edge is sampled from $E_{j}(v, K_{\ell})$ is
  at least $\beta /(2h|\K^i(v)|{n^{\epsilon}}) \geq 1 / (2h\alpha{n^{\epsilon}})$.
  Letting $Z_{\ell}$ be the number of query edges in $E_{j}(v, K_{\ell})$ created at
  the $i$-th trial, for any $1 \leq \ell \leq q$, we have
  \[
    \Pr[Z_{\ell} = 0] \leq \left(1 -
    \frac{1}{2h\alpha{n^{\epsilon}}}\right)
    ^{c^2\exp_{n}(2^{j}\delta + \epsilon)\log^3 n }
    \leq \exp_{e}\left(-\frac{c^2\exp_n(2^{j}\delta)\log^3 n }{2h\alpha}\right)
    \leq e^{-c\log n /6}.
  \]
  Thus, in every trial, at least one node in each group $K_{\ell}$ is queried
  by $v$ whp. If $|\K^i(v)| \leq \alpha$ holds, $\beta = 1$ holds and thus
  each group consists of a single node in $\K^i(v)$. Thus all nodes
  in $\K^i(v)$ are queried by $v$ whp. in the $i$-th trial (note
  that no node becomes a residual in the case of $\beta = 1$). Otherwise, $q =
  \lfloor |\K^i(v)|/\beta \rfloor \geq \lfloor (\beta - 1)\alpha/\beta \rfloor
  \geq \alpha/3 = (c \exp_n(2^{j}\delta) \log^2 n) /h \geq c \exp_n(2^{j}\delta)
  \log n$ holds, and thus
  $v$ queries at least $c \exp_n(2^{j}\delta) \log n$ nodes in the $i$-th trial.
  Consequently, if $|\K^i(v)| \leq \alpha$ holds for all $1 \leq i \leq 2h$, $v$
  queries all nodes in $\K^i(v)$ at the $i$-th trial, that is,
  queries all nodes in $N_j(v)$ throughout the run of \Cluster{j}. 
  Then $v$ becomes light. If $|\K^i(v)| > \alpha$ holds for some $i$,
  $v$ queries at least $c \exp_n(2^{j}\delta) \log n$ nodes in $\K^i(v)$,
  which implies $v$ becomes heavy or light.

  Finally, let us show that any node $v \in V_k$ is light. Since
  $n_k \leq 3\exp_{n}(1 - (2^{k} - 1)\delta)/2 = 3\exp_{n}(2^k\delta)/2 \leq
  (3c\exp_{n}(2^{k}\delta)\log^2 n)/h = \alpha$ holds from Lemma
  \ref{lem:bound-graph-size}, we have $|N_{j}(v)| \leq n^k \leq \alpha$.
  Since $\K^i(v)$ is a subset of $N_j(v)$, $|\K^i(v)| \leq \alpha$ holds
  for all $i$. By the argument above, then $v$ is light.
\end{proof}

The rest of the analysis is divided into two parts:
First, in Section~\ref{sec:bounding-stretch}, we analyze the stretch of $H$,
proving that it is at most $\kappa = O(3^k)$.
Section~\ref{sec:bound-numb-edges} then establishes an
$\tilde{O} (n^{1 + \delta})$
upper bound on the number
of edges in $H$.

%%%%%%%%%%%%%%%%%%%%%%%%%%%%%%%%%%%%%%%
\subsection{Bounding the Stretch}
\label{sec:bounding-stretch}
%%%%%%%%%%%%%%%%%%%%%%%%%%%%%%%%%%%%%%%

The following lemma is a well-known fact.
\begin{lemma}[\cite{peleg1989graph}] \label{lem:spanner}
Let $H = (V, X)$ be any (spanning) subgraph of $G=(V, E)$ and
$\C$ be the partition of $V$ such that for any $C \in \mathcal{C}$,
$\Ind_H(C)$ has a diameter at most $\ell$. If $X$ contains
at least one edge in $E(C_i, C_j)$ for any pair $(C_i, C_j) \in \mathcal{C}^2$
such that $E(C_i, C_j)$ is nonempty, $H$ is a $(2\ell + 1)$-spanner.
\end{lemma}

\begin{comment}
\begin{proof}
For the proof, it suffices to show that for any edge $e = (v, u) \in E$
the distance between $v$ and $u$ in $H$ is at most $(2l + 1)$.  If
there exists $C \in \C$ such that $\Ind_G(C)$ contains $e \in E$,
$\Dist_H(u, v) \leq l$ because the diameter of $\Ind_H(C)$ is at most $l$.
Otherwise, $u$ and $v$ belong to two different subsets in $\C$, which
are respectively denoted by $C_u$ and $C_v$. Let $(u', v')$ be an arbitrary
edge in $X \cap E(C_u, C_v)$ ($u' \in C_u$ and $v' \in C_v$).
Then we have $\Dist_H(u, v) \leq \Dist_H(u, u') + 1  + \Dist_H(v', v) \leq 2l + 1$.
The lemma holds.
\end{proof}
\end{comment}

Let $V'_j \subseteq V_j$ ($1 \leq j \leq k - 1$) be
the set of the nodes unclustered in the run of \Cluster{j},
and $V' = \bigcup_{1 \leq j \leq k - 1} V'_j$. We define
$C_j(v) \subseteq V$ as the set of nodes in $V$ which are
clustered into $v \in V_j$, and also define $r(v)$ as the value
$j$ satisfying $v \in V'_j$ for any $v \in V'$.
Let $C(v) = C_{r(v)}(v)$ for short.

\begin{lemma} \label{lem:part-diam}
Let $H = (V, S)$ be the (spanning) subgraph output by {\Sampler}.
The diameter of $\Ind_H(C_j(v))$ for any $v \in V_j$ is at
most $r = 3^j - 1$.
\end{lemma}

\begin{proof}
We show that $\Ind_H(C_j(v))$ contains a spanning tree of $\Ind_G(C_j(v))$ with height
at most $3^j - 1$. The proof follows the induction on $j$.
For $j = 0$, $\Ind_H(C_j(v)) = \Ind_G(C_j(v))$ is a graph consisting of
a single node, and thus its diameter is zero. Suppose as the induction hypothesis
that the lemma holds for some $j$, and consider the case of $j+1$. Since any node
$v \in V_{j+1}$ corresponds to a center node $v \in V_j$ and a star-based
connection with its neighbors in $V_j$, $\Ind_H(C_j(v))$ is obviously contains
a spanning tree of $\Ind_G(C_j(v))$ with a diameter at most
$3(3^j - 1) + 2 = 3^{j+1} - 1$. The lemma follows.
\end{proof}

Finally the following theorem is deduced.
\begin{theorem}
The graph
$H = (V, S)$ output by {\Sampler} is an $O(2\cdot3^k - 1)$-spanner of
$G$ whp.
\end{theorem}

\begin{proof}
%Note that $C(v) = C_j(v)$ holds for $v \in V'_j$.
Since any node in $G_k$ is unclustered,
$\mathcal{C} = \{C(v) \mid v \in V'\}$ is a partition of $V$.
If $C(u)$ and $C(v)$ ($u, v \in V'$) are neighboring and
$r(u) \leq r(v)$ holds, there exists a node $w \in V_{r(u)}$ such that
$u$ and $w$ are neighboring in $G_{r(u)}$ and $C_{r(u)}(w) \subseteq
C(v)$ holds. Since every unclustered node is light (by Lemmas~\ref{lem:heavynode}
and \ref{lem:unclustered}), $C(u)$ is light. Then
at least one edge in $E(C_{r(u)}(u), C_{r(u)}(w))$ is added to $S$, which
implies that at least one edge in $E(C(u), C(v))$ is added to $S$.
Consequently, the edge set $S$ constructed by \Sampler{} satisfies the condition of
Lemma~\ref{lem:spanner} w.r.t. $\C$. The remaining issue is to bound
the diameter of $C(v) \in \mathcal{C}$ for all $v \in V'$, which is shown
by Lemma~\ref{lem:part-diam}.
\end{proof}

%%%%%%%%%%%%%%%%%%%%%%%%%%%%%%%%%%%%%%%
\subsection{Bounding the Number of Edges}
\label{sec:bound-numb-edges}
%%%%%%%%%%%%%%%%%%%%%%%%%%%%%%%%%%%%%%%

Using Lemma~\ref{lem:bound-graph-size}, we can bound the number of edges in $S$
output by \Sampler{}.

\begin{lemma}
The output edge set $S$ contains $\tilde{O}(n^{1 + \delta})$ edges.
\end{lemma}

\begin{proof}
Each trial of the first step in the run of \Cluster{j}
adds $O(\exp_{n}(2^j\delta)\log^3 n)$ edges to $S$ per node
in $V_{j}$, and $n_j = O(\exp_{n}(1 - (2^j - 1)\delta))$ holds by
Lemma~\ref{lem:bound-graph-size}. Then the total number of edges added to $S$
in \Cluster{j} is $h \cdot O(\exp_{n}(2^j\delta)\log^3 n) \cdot
n_j = O(hn^{1 + \delta}\log^3 n)$, and thus the size of $S$ is
$O(khn^{1 + \delta}\log^3 n)$.  Since $k, h \leq \log n$,
we obtain the lemma by omitting all the logarithmic factors.
\end{proof}

%%%%%%%%%%%%%%%%%%%%%%%%%%%%%%%%%%%%%%%%%%%%%%%%%%%%%%%%%%%%%%%%%%%%%%%%%%%%%%
\section{Distributed Implementation}
\label{sec:distr-impl}
%%%%%%%%%%%%%%%%%%%%%%%%%%%%%%%%%%%%%%%%%%%%%%%%%%%%%%%%%%%%%%%%%%%%%%%%%%%%%%
In this section, we explain how the centralized algorithm presented in
Section~\ref{sec:algorithm} is implemented in a distributed fashion over the
(simple) communication graph
$G = (V, E)$
with round complexity
$O (3^kh)$
and message complexity
$\tilde{O} (n^{1 + 1 / (2^{k+1} - 1) + (1/h)}) =
\tilde{O} (n^{1 + \delta + \epsilon})$
whp.
\begin{comment}
We also show that when the algorithm terminates, the nodes hold enough
information to route a message from node $u$ to node $v$ over a path in $H$ of
length at most $\kappa$ for every
$(u, v) \in E$.
\end{comment}

The key observation in this regard is that procedures \Cluster{j}
would have been naturally distributed if the nodes in
$V_{0}, V_{1}, \dots, V_{k}$
could have performed local computations and exchanged messages over their
incident edges in
$G_{0}, G_{1}, \dots, G_{k}$,
respectively (recall that graphs
$G_{1}, \dots, G_{k}$
are virtual, defined only for the sake of the algorithm's presentation).
Indeed, the action of marking a node as a center and the action of marking an
edge as a query/probe edge are completely local and do not require any
communication (in $G_{j}$).
The action of checking whether a query edge $e$ leads to a center and the
action of identifying all the edges parallel to a query/probe edge $e$ are
easily implemented under the LOCAL model with unique edge IDs by sending a
constant number of messages over edge $e$ in $G_{j}$.

So it remains to explain how local actions in graph $G_{j}$,
$1 \leq j \leq k$,
are simulated in the actual communication graph $G$.
Let $T_j(v)$ be the spanning tree of $\Ind_G(C_j(v))$ shown in the proof of
Lemma~\ref{lem:part-diam}. Once a cluster
$C_j(v)$ is formed, %shay: was: organized
no further
edge is added to the inside of the cluster. Thus $T_j(v)$ is already contained
in $C_j(v)$ at the beginning of the run of \Cluster{j}. In the distributed
implementation, the local actions of node $v$ in $G_{j}$ are simulated by nodes $C_j(v)$ in $G$
via a constant number of broadcast-convergecast sessions over $T_j(v)$ rooted
at $v \in V$. (This is made possible by the choice of the LOCAL model
with unique edge IDs). This process requires sending
$O (1)$ (additional) messages over each edge in $T_j(v)$, and by
Lemma~\ref{lem:part-diam}, it takes at most $O(3^j)$ rounds.

\sloppy
\begin{theorem}
Algorithm \Sampler{} has round complexity
$O (3^k h)$
and message complexity
$\tilde{O} (n^{1 + \delta + \epsilon})$
whp.
\end{theorem}
\par\fussy
\begin{proof}
In the first step of \Cluster{j}, each trial takes $O(1)$
rounds in $G_{j}$. The second step takes $O(1)$ rounds in $G_{j}$.
Hence the total running time of \Cluster{j} takes $O(h(3^j - 1))$ rounds in $G$.
Summing up it over all $0 \leq j \leq k$, the bound on the round complexity is
$\sum_{j = 0}^{k} O(h(3^j - 1)) = O (3^kh)$.

For the message complexity, the simulation of one round in $G_j$ is
implemented with an additive overhead incurred by
a constant-number sessions of broadcast and convergecast in each cluster
$C(v)$ for $v \in V_j$, which use $O(n)$ messages in total.
Each trial of the first step in \Cluster{j} uses
$O(\exp_{n}(2^j\delta+ \epsilon)\log^3 n)$
messages per node in $V_{j}$. Thus the total message complexity in
\Cluster{j} is $O(\exp_{n}(2^j\delta + \epsilon)\log^3 n) \cdot
n_j = O(hn^{1 + \delta + \epsilon}\log^3 n)$ by
Lemma~\ref{lem:bound-graph-size}. Summing up this over $0 \leq j \leq k$,
we can conclude that the message complexity is
$O(khn^{1 + \delta + \epsilon}\log^3 n ) = \tilde{O}(n^{1 + \delta + \epsilon})$
(recall $k, h \leq \log n$).
\end{proof}

\begin{comment}
For a local routing scheme that enables nodes
$x, y \in V$
to identify an $(x, y)$-path in $H$ of length at most $\kappa$ for every
$e = (x, y) \in E - S$,
we simply modify the algorithm so that nodes $x$ and $y$ keep record of the
path $P$ promised in Lemma~\ref{lem:bound-stretch}.
Iterating the proof of Lemma~\ref{lem:bound-stretch} with this goal in mind
and using the infrastructure developed in the current section, we observe that
it can be done by sending $O (1)$ additional messages over each edge in $S$
for every iteration of \Sampler{} and does not affect the (asymptotic) round
and message complexities.
\end{comment}

%%%%%%%%%%%%%%%%%%%%%%%%%%%%%%%%%%%%%%%%%%%%%%%%%%%%%%%%%%%%%%%%%%%%%%%%%%%%%%
\section{Message-Efficient Simulation of Local Algorithms}
\label{sec:simulation}
%%%%%%%%%%%%%%%%%%%%%%%%%%%%%%%%%%%%%%%%%%%%%%%%%%%%%%%%%%%%%%%%%%%%%%%%%%%%%%

In this section, we provide a new and versatile
message-reduction scheme for LOCAL algorithms based on the new
\Sampler{} algorithm.
The technical ingredients of this scheme consist of a message-efficient $t$-local
broadcast
algorithm built on top of constructed spanners, which is commonly used
in the past message-reduction schemes~\cite{Censor-HillelHKM2012,Haeupler2015}.

Consider the initial configuration where each node $v \in V$ has a message $M_v$,
and let $B_{G,t}(v) = \{u \mid \Dist_G(v, u) \leq t\}$.
The task of the $t$-local broadcast is that each $v \in V$ delivers
%shay:was informs M_v, which means tells something to M_v
 $M_v$ to all the
nodes $u \in B_{G,t}(v)$. In any $t$-round LOCAL algorithm, the computation at
node $v \in V$ relies only on the initial knowledge (i.e., its ID, initial state,
and incident edge set) of the nodes in $B_{G,t}(v)$, and thus any $t$-local
broadcast algorithm in the LOCAL model can simulate any $t$-round LOCAL algorithm.
The core of the scheme is the following theorem.

\begin{lemma} \label{lem:simulation}
There exist two $t$-local broadcast algorithms respectively achieving
the following time and message complexities:
\begin{itemize}
  \item $\tilde{O}(tn^{1 + 2/(2^{\gamma + 1} - 1)})$ message complexity and
  $O(3^\gamma t + 6^\gamma)$ round complexity for any
  $1 \leq \gamma \leq \log\log n$ and $t \geq 1$,
  \item $\tilde{O}(t^2n^{1 + 1/O(t^{1/\log\log t})\log t)})$ message
  complexity and $O(t)$ round complexity for $t \geq 1$.

\end{itemize}
\end{lemma}

\begin{proof}
Consider the realization of the first algorithm. For any $v$, all the nodes
in $B_{G, t}(v)$ are within $\alpha t$-hop away
from $v$ in any $\alpha$-spanner. Thus, once we got any $\alpha$-spanner
$H = (V, S)$, the local flooding within distance $\alpha t$ in $H$
trivially implements $t$-local broadcast. Setting
$k = \gamma$ and $h = (2^{\gamma + 1} -1)$ of Theorem~\ref{thm:main}
implements the spanner satisfying the first condition, where the additive
$O(6^\gamma)$ term is the time for spanner construction.
For the second algorithm, we utilize the spanner-construction algorithm
by Derbel et al.\cite{DerbelGPV2009} which provides
a $(3, O(3^k))$-spanner $H$ with
$\tilde{O}(3^kn^{1 + 1/O(k)})$ edges within $O(3^k)$ rounds for any $k \geq 1$.
Consider the algorithm by Derbel et al. with parameter
$k = \lceil \log_3 t - \log\log_3 t \rceil$, which results in the
$O(t/\log_3 t)$-round algorithm of constructing $(3, O(t))$-spanner with
$\tilde{O}(tn^{1 + 1/O(\log t)})$ edges. We run this algorithm on top
of the first simulation scheme with parameter $\gamma = \log_3 \log_3 t$.
The simulated algorithm constructs a
$(3, O(t))$-spanner $H'$ with $\tilde{O}(tn^{1 + 1/O(\log t)})$ edges
spending $O(3^{\log_3 \log_3 t} \cdot t/\log_3 t + 6^{\log_3 \log_3 t}) = O(t)$ rounds and
$\tilde{O}(tn^{1 + 1/O(\log t)})$ messages.
The local flooding within distance $3t + O(t)$ on top of $H'$ implements the
$t$-local broadcast, which takes $O(t)$ rounds and
$\tilde{O}(t^2n^{1 + 1/ O(\log t)})$ messages. The lemma is proved.
\end{proof}
As stated above, Theorem~\ref{thm:simulation} is trivially deduced from
Lemma~\ref{lem:simulation}.

%%%%%%%%%%%%%%%%%%%%%%%%%%%%%%%%%%%%%%%%%%%%%%%%%%%%%%%%%%%%%%%%%%%%%%%%%%%%%%
\section{Concluding Remarks}
\label{sec:concludingremarks}
%%%%%%%%%%%%%%%%%%%%%%%%%%%%%%%%%%%%%%%%%%%%%%%%%%%%%%%%%%%%%%%%%%%%%%%%%%%%%%

In this paper, we present an efficient spanner construction as well as two message-reduction schemes for LOCAL algorithms
that preserve the asymptotic time
complexity of the original algorithm. The reduced message complexity is close
to linear (in $n$). Is this the best possible in constructing a spanner?
Similarly, some open questions still lie on the line of
developing efficient message-reduction schemes: (1) While
our scheme only sends $\tilde{O}(t^2n^{1 + O(1/\log t)})$ messages for simulating
$t$-round algorithms, it is not clear whether the additive $O(1/\log t)$ term in
the exponent can be improved further. Can one have a message-reduction
scheme with $\tilde{O}(\mathrm{poly}(t)n^{1 + o(1/\log t)})$ message complexity and
no overhead in the round complexity? (2)  Algorithm \Sampler{} inherently relies on randomized
techniques for probing the neighbors in $G_j$ using only few messages.
Is it possible to obtain a deterministic message-reduction scheme with no
degradation of time?

Very recently, the authors received a comment on the first question,
which states that utilizing the spanner construction by Elkin and 
Neiman~\cite{ElkinN2017} will improve the message complexity. 
Unfortunately, due to lack of time, we do not completely check this 
idea, and thus the current version only states the result based on the 
algorithm by Derbel et al., but certainly it is a promising approach.
If it actually works, the message complexity will be reduced to 
$\tilde{O}(t^2n^{1 + O(1/t^{1/\log\log t})})$. 

Finally, we note that using an $o(m)$ messages spanner construction algorithm that does not increase the time can be useful also for global algorithms in the LOCAL model. It implies that any function can now be computed on the graph in strictly optimal $O(\mathrm{diameter})$ time 
and $o(m)$ messages (for large enough $m$).

\end{document}